\documentclass[aps,prx,twocolumn,superscriptaddress,nofootinbib]{revtex4-2}

\usepackage{amsmath,amssymb,amsfonts,amsthm}
\usepackage{braket}
\usepackage{bm}
\usepackage{xcolor}
\usepackage{tikz}
\usepackage[normalem]{ulem}

\usepackage{algorithm}
\usepackage{algpseudocode}
\usepackage[colorlinks=true,citecolor=magenta,urlcolor=cyan,linkcolor=red]{hyperref}

\makeatletter
\newcommand\fs@doubleruled{%
  \def\@fs@cfont{\bfseries}%
  \let\@fs@capt\floatc@ruled
  \def\@fs@pre{%
    \hrule height 1.2pt depth 0pt
    \kern 2pt
  }%
  \def\@fs@mid{%
    \kern 2pt
    \hrule height 0.5pt depth 0pt
    \kern 1.2pt
    \hrule height 0.5pt depth 0pt
    \kern 2pt
  }%
  \def\@fs@post{%
    \kern 2pt
    \hrule height 1.2pt depth 0pt
  }%
  \let\@fs@iftopcapt\iftrue
}
\makeatother

\restylefloat{algorithm}

\usetikzlibrary{positioning, arrows.meta, shapes.geometric}
\newtheorem{theorem}{Theorem}
\newtheorem{lemma}{Lemma}
\newtheorem{proposition}{Proposition}

\newtheorem{corollary}{Corollary}

\usepackage{physics}

\newcommand{\id}{\mathbb I}
\newcommand{\ii}{\mathrm i}

\usepackage{mathtools}

\usetikzlibrary{shadows.blur, shadings, arrows.meta, positioning, calc}

\begin{document}

\title{Distributed Trotterization with optimal time-scaling entanglement cost}

\author{Tianfeng Feng}
\email{feng.tianfeng.phys@gmail.com}
\affiliation{QICI Quantum Information and Computation Initiative, Department of Computer Science, The University of Hong Kong, Pokfulam Road, Hong Kong SAR, China}
\author{Jinzhao Sun}
\email{jinzhao.sun.phys@gmail.com}
\affiliation{School of Physical and Chemical Sciences, Queen Mary University of London, London E1 4NS, United Kingdom}
 
\author{Yunlong Xiao}
\email{mathxiao123@gmail.com}
\affiliation{Institute of Advanced Intelligence and Computing (IAIC), Agency for Science, Technology and Research (A*STAR), 1 Fusionopolis Way, \#16-16 Connexis, Singapore 138632, Republic of Singapore}
\author{Qi Zhao}
\email{zhaoqi@cs.hku.hk}
\affiliation{QICI Quantum Information and Computation Initiative, Department of Computer Science, The University of Hong Kong, Pokfulam Road, Hong Kong SAR, China}

\date{\today}

\begin{abstract}
Distributed architectures extend quantum simulation of many-body dynamics beyond the reach of any single processor, with shared entanglement mediating interactions between spatially separated devices. 
Conventional implementations rely on quantum teleportation, which provides a universal realization of nonlocal operations but incurs a fixed entanglement cost per gate, irrespective of its strength. 
This becomes increasingly inefficient in product formula simulation, where higher accuracy requires ever more numerous, yet progressively weaker, nonlocal rotations, causing the entanglement cost to diverge in the high-accuracy limit.
Here we introduce a simple repeat-until-success protocol that makes entanglement consumption adaptive to interaction strength. 
Incorporating this primitive into distributed product formulas yields a total entanglement cost that scales linearly with evolution time and independent of Trotter error.
A matching lower bound from quantum communication complexity proves this time scaling to be optimal, establishing a foundation for resource-efficient high-accuracy quantum simulation across networked processors.

\end{abstract}


\maketitle

\noindent \textbf{Introduction}--Some of the most powerful advances in computing have come not from improving individual processors alone, but from enabling many processors to work together.
Distributed quantum computing extends this approach to quantum technologies that connect spatially separated processors using classical communication and shared entanglement, allowing quantum information held across a network to be processed collectively \cite{buhrmanDistributedQuantumComputing2003,Caleffi_2024}.
Such architectures offer a natural route towards computation beyond the capabilities of any single device~\cite{bealsEfficientDistributedQuantum2013, ainleyMultipartiteEntanglementMultinode2024,kimbleQuantumInternet2008,elkinCanQuantumCommunication2014,wehnerQuantumInternetVision2018, legallQuantumAdvantageLOCAL2019,caleffiDistributedQuantumComputing2024,cacciapuotiMultipartiteEntanglementDistribution2024}, with proposed applications spanning factoring~\cite{yimsiriwattanaDistributedQuantumComputing2004, jiangDistributedShorAlgorithm2023}, machine learning~\cite{tangCommunicationefficientQuantumAlgorithm2023, liBlindQuantumMachine2024}, sensing~\cite{guoDistributedQuantumSensing2020}, phase estimation~\cite{liuDistributedQuantumPhase2021} and error correction~\cite{xuDistributedQuantumError2022}, among others~\cite{tanDistributedQuantumAlgorithm2022,anshuDistributedQuantumInner2022,montanaroQuantumCommunicationComplexity2024,chen2025distributed,llovo2025network,11587733,sun2021perturbative}.
Distributed quantum simulation is particularly important in this setting~\cite{buessenSimulatingTimeEvolution2023,feng2024distributedquantumsimulation,russo2026cosma_dqs}. 
As quantum many-body systems grow beyond the capacity of individual processors, their degrees of freedom must be distributed across multiple devices, and their dynamics inevitably involve interactions between spatially separated subsystems. 
Shared entanglement therefore becomes a central resource for reproducing such dynamics across a quantum network.

Product formulas provide a natural setting in which to expose this resource cost by decomposing the target evolution $e^{-\ii Ht}$ over time $t$ into a sequence of short-time local and nonlocal operations.
For a $q$th-order product formula, achieving simulation error $\epsilon$ requires $r=\mathcal{O}\left(t^{1+1/q}\epsilon^{-1/q}\right)$ Trotter steps~\cite{lloydUniversalQuantumSimulators1996, childsNearlyOptimalLattice2019,childsTheoryTrotterError2021}, and therefore an increasing number of nonlocal gates as the approximation is refined.
A direct way to implement these gates is through quantum teleportation: the relevant subsystem is transferred to one processor, evolved locally, and teleported back, at a fixed cost of 2 ebits per nonlocal operation \cite{bennett1993teleporting}. 
The total entanglement cost therefore grows with the number of Trotter steps, i.e., $\mathcal{O}\left(t^{1+1/q}\epsilon^{-1/q}\right)$ \cite{feng2024distributedquantumsimulation}. 
This scaling, however, misses a basic feature of product formulas. 
Finer Trotterization introduces more nonlocal rotations, but each rotation becomes correspondingly weaker. 
Teleportation assigns the same entanglement cost to every rotation, irrespective of its strength, and therefore pays an increasingly large overhead as higher accuracy is sought. 
The key question is thus whether distributed simulation can instead consume entanglement in proportion to the strength of the nonlocal dynamics being realized.

In this work, we answer the question affirmatively by introducing a repeat-until-success (RUS) protocol for the nonlocal rotations that form the basic building blocks of distributed Trotterization. 
Rather than consuming a maximally entangled Bell pair, as in teleportation, each attempt uses a weakly entangled state whose entanglement is matched to the strength of the target rotation.
Each round succeeds with probability 1/2; 
upon failure, the rotation angle is doubled and a stronger entangled resource is used in the next round \cite{Cirac_2001RUS,RUS2013}. 
Although later rounds consume more entanglement, they are reached with exponentially decreasing probability, leaving the average cost of a weak rotation proportional to its strength.
When applied across the full simulation, this adaptive structure of the shared weakly entangled state removes the dependence on the number of Trotter steps: 
the total entanglement cost scales as $\mathcal{O}(t)$, growing only linearly with the evolution time and remaining independent of the target simulation error.
This is fundamentally different from teleportation, where every nonlocal gate carries a fixed cost, so finer Trotterization steadily increases the resource demand and can make it diverge as higher accuracy is required.
Moreover, lower bounds from quantum communication complexity~\cite{feng2024distributedquantumsimulation} show that the linear dependence on evolution time is optimal in general, establishing the RUS construction as an optimal protocol for distributed quantum simulation.

\begin{figure}[t]
    \centering
    \includegraphics[width=1\linewidth]{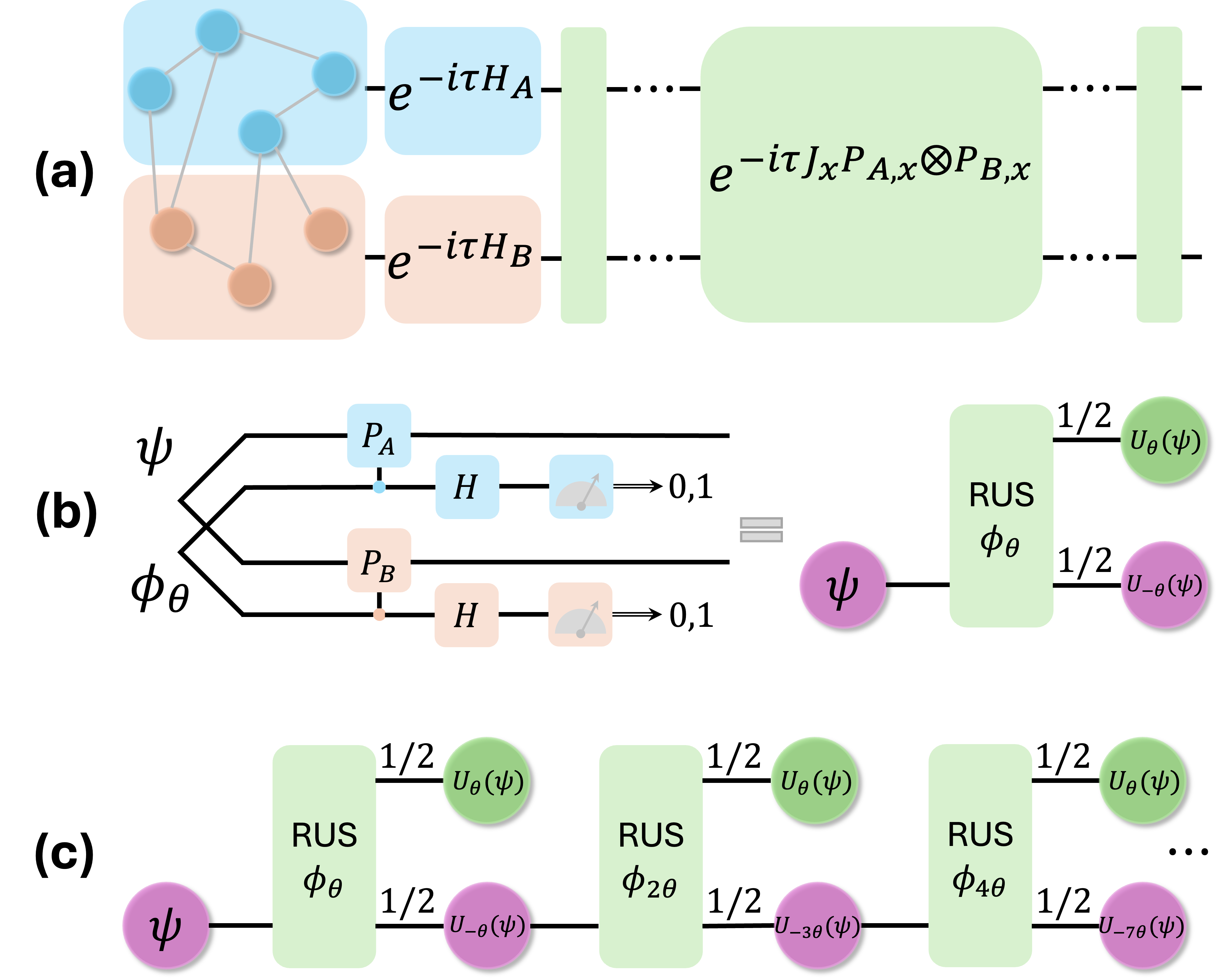}
    \caption{\textbf{Distributed Trotterization}.
    (a) Schematic of the Distributed Quantum Simulation. 
    A many-body Hamiltonian distributed across spatially separated quantum processors. 
    The system is partitioned into subsystems $A$ and $B$, whose local dynamics are generated by $H_A$ and $H_B$, respectively, while the interaction term $H_{AB}$ is implemented through entangling evolutions.
    (b) Repeat-Until-Success Protocol.
    The ancillary state $\phi_\theta$, together with LOCC, implements the target evolution $U_{\theta}$ probabilistically: 
    the even- and odd-parity outcomes occur with equal probability, yielding $U_{\theta}(\psi)$ and $U_{-\theta}(\psi)$, respectively.
    (c) Adaptive Quantum Simulation.
    If a round returns the odd-parity outcome, the protocol is repeated with a doubled interaction angle. 
    After $j$ consecutive failures, the resource state $\phi_{2^j\theta}$ is used, correcting the accumulated inverse evolution and yielding the target state $U_{\theta}(\psi)$ with probability 1/2 at each round.
    }
    \label{fig:Sketch}
\end{figure}

\noindent \textbf{Distributed Simulation}---Distributed quantum simulation (DQS) enables dynamics beyond the capabilities of a single processor by coordinating spatially separated devices~\cite{buessenSimulatingTimeEvolution2023,feng2024distributedquantumsimulation,russo2026cosma_dqs}. 
Shared entanglement links these remote nodes and provides the resource required to mediate couplings between them. 
To capture this structure, we consider an arbitrary bipartition of the network into two parties $A$ and $B$ (see Fig.~\ref{fig:Sketch}(a)), with Hamiltonian $H=H_A+H_B+H_{AB}$, where $H_A$ and $H_B$ govern the dynamics internal to each party, while $H_{AB}$ contains the terms that connect them.
Expanding the latter in Pauli strings $P_{A,x}$ and $P_{B,x}$, we write
\begin{align}\label{eq:H_AB}
    H_{AB}
    \coloneqq
    \sum_{x\in\mathfrak{X}}
    J_x P_{A,x}\otimes P_{B,x}.
\end{align}
The coefficient $J_x$ sets the coupling strength, and $\mathfrak{X}$ labels the cross-partition interactions.
Since the evolutions generated by $H_A$ and $H_B$ can be implemented locally, they require no shared entanglement. 
We therefore treat these dynamics as free, leaving $H_{AB}$ as the sole source of entanglement cost.

To approximate the dynamics generated by $H$ over a total time $t$, we partition the evolution into $r$ steps of duration $\tau\coloneqq t/r$. 
A $q$th-order product formula then gives~\cite{suzukiGeneralTheoryFractal1991,childsTheoryTrotterError2021}
\begin{align}\label{eq:Suzuki-Trotter}
    e^{-iHt}=\left(S_q(\tau)\right)^r+
    \mathcal O(t^{q+1}/r^q),
\end{align}
where $S_q(\tau)$ denotes a single simulation step.  
Increasing $r$ therefore reduces the step duration and systematically improves the approximation. 
For the interaction Hamiltonian $H_{AB}$ in Eq.~\eqref{eq:H_AB}, the nonlocal content can be organized into a finite number $\nu_q$ of sub-layers interleaved with local evolutions using  product formula \cite{childsTheoryTrotterError2021},
\begin{equation}
    S_q(\tau)=\text{local gates}
    \,\times\,
    \left(
    \prod_{\alpha=1}^{\nu_q}
    \prod_{x\in\mathfrak{X}_\alpha}
    e^{
        -i\tau c_{\alpha,x} J_x  
        P_{A,x}\otimes P_{B,x}
    }
    \right).
    \label{eq:general_product_formula_segment}
\end{equation}
Here, $\mathfrak{X}_\alpha$ collects the interaction terms appearing in sub-layer $\alpha$, while the coefficients $c_{\alpha,x}$, determined by the chosen product formula, satisfy $|c_{\alpha,x}|\leqslant 1$. 
At fixed order $q$, both $\nu_q$ and $c_{\alpha,x}$ are independent of the evolution time $t$ and the number of Trotter steps $r$. 
Taking $\nu_1=1$, $\mathfrak{X}_1=\mathfrak{X}$, and $c_{1,x}=1$ recovers the first-order Trotter decomposition (see Fig.~\ref{fig:Sketch}(a)), with $S_1(\tau)=e^{-i\tau H_A}e^{-i\tau H_B}\prod_{x\in\mathfrak{X}}e^{-i\tau J_x P_{A,x}\otimes P_{B,x}}$~\cite{lloydUniversalQuantumSimulators1996}.

The product formula decomposition expresses simulations via elementary Pauli rotations of the form $U_\theta\coloneqq e^{-i\theta P_A\otimes P_B}$, where, for example, $\theta=t c_{\alpha,x} J_x/r$ (see Eq.~\eqref{eq:general_product_formula_segment}), so that the entanglement cost of DQS is ultimately governed by that of these nonlocal operations, particularly in the small $\theta$ regime.
For a bipartite pure state $\psi\coloneqq\ketbra{\psi}{\psi}$, its entanglement is measured by the entropy $E(\psi)\coloneqq -\mathrm{Tr}[\rho_A\log_2\rho_A]$~\cite{nielsen2000quantum}, with $\rho_A\coloneqq\Tr_{B}[\psi]$. 
This notion extends naturally to a unitary through its entangling power $E_{\mathrm{gen}}(U)=\sup_{\psi\,\mathrm{Sep}} E(U(\psi))$~\cite{nielsen2000quantum,soeda2011entanglement,chen2016entanglement},
which captures the maximum entanglement that $U$ can generate from a separable input. 
Since local operations and classical communication (LOCC) cannot create entanglement, any distributed realization of $U$ must consume at least $E_{\mathrm{gen}}(U)$ ebits of shared entanglement. 
For $U_\theta$, this bound takes the form $E_{\mathrm{gen}}(U_\theta) =h_2(\sin^2\theta)= \theta^2\log_2(1/\theta^2) + \mathcal{O}(\theta^2)$ as $\theta\to 0$, where $h_2$ denotes the binary entropy.
Thus, the intrinsic entanglement of a weak rotation vanishes with its strength, whereas teleportation continues to consume a fixed amount per nonlocal gate. 
This makes teleportation increasingly inefficient in the weak rotation regime.

\noindent \textbf{Repeat-Until-Success Protocol}---A straightforward implementation of the distributed rotation $U_\theta$ is to teleport the relevant qubits to one party, apply the evolution locally, and teleport them back. 
Although universal, this approach consumes 2 ebits per rotation, irrespective of $\theta$, and thus becomes increasingly wasteful as the rotation weakens.
Repeat-until-success (RUS) constructions provide an alternative route to realizing quantum operations~\cite{Cirac_2001RUS,RUS2013,sun2026quantum}. 
Here, we introduce an RUS protocol tailored to distributed Pauli rotation $U_\theta$ of multiple qubits, consisting of the following steps (see Fig.~\ref{fig:Sketch}(b)):

\begin{enumerate}
\item {Resource Preparation.} The protocol begins with Alice and Bob preparing a shared weakly entangled state $\ket{\phi_\theta} \coloneqq \cos\theta\ket{00} - i\sin\theta\ket{11}$.
\item {Controlled Operations.} Alice and Bob then apply local controlled-Pauli operations $P_A$ and $P_B$ using their respective halves of the shared state $\ket{\phi_\theta}$ as controls and the corresponding subsystems of $\ket{\psi}$ as targets.
\item {Parity Check.} The shared qubits are subsequently measured in the Hadamard basis, with outcomes $m_A$ and $m_B$ communicated between the parties.
Even parity realizes the desired rotation $U_\theta$ on $\ket{\psi}$, whereas odd parity yields the inverse rotation $U_{-\theta}$; 
each branch occurs with probability $1/2$.
\item {Recursive Correction.} An odd-parity outcome in step 3 leaves the data state as $U_{-\theta}\ket{\psi}$. 
Alice and Bob then prepare $\ket{\phi_{2\theta}}$ and repeat the protocol, yielding $U_{\theta}\ket{\psi}$ with probability 1/2 or $U_{-3\theta}\ket{\psi}$ with probability 1/2, as illustrated in Fig.~\ref{fig:Sketch}(c). 
After each subsequent failure, the rotation angle is doubled and the procedure repeated until the cumulative success probability reaches the desired threshold $1-\delta$.
\end{enumerate}

We next quantify the entanglement cost of the RUS protocol. 
The first round consumes the state $\ket{\phi_\theta}$, with entanglement $E_1\coloneqq E(\phi_\theta)$. 
A failure, which occurs with probability 1/2, causes the protocol to proceed to the next round, where $\ket{\phi_{2\theta}}$ is consumed (see Fig.~\ref{fig:Sketch}(c)).
More generally, round $j$ uses $\ket{\phi_{2^{j-1}\theta}}$, with entanglement $E_j\coloneqq E(\phi_{2^{j-1}\theta})$, and is reached with probability $p_{j}\coloneqq2^{-(j-1)}$.
The expected entanglement cost of a protocol truncated after $K$ rounds is therefore
\begin{align}
    C_K(\theta)
    \coloneqq
    \sum_{j=1}^{K}p_{j}E_j
    =
    \sum_{j=0}^{K-1} 2^{-j}\,
    h_2\!\left(\sin^2(2^j\theta)\right),
\end{align}
where the second equality follows from the relabeling $j-1\to j$.
Although the rotation angle doubles after each failure, the probability of reaching successive rounds falls exponentially. 
This balance yields a simple upper bound on the total entanglement cost that is linear in $\theta$; 
namely,

\begin{lemma}[{\bf Entanglement Cost of Implementing $U_\theta$}]
\label{lem:per_gate}
For any $K\geqslant 1$ and $|\theta|\leqslant 1$, the entanglement cost of the RUS protocol truncated after $K$ rounds obeys $C_K(\theta) \leqslant 9|\theta|$.
\end{lemma}

The proof and further analysis are provided in Appendix~\ref{app:longtime_rus_cost}, where the additional result for the proposed RUS protocol is established: the entanglement cost itself scales linearly with $|\theta|$, i.e., $C_K(\theta) = \Theta(|\theta|)$, rather than merely admitting a linear upper bound. 
For a simulation using $q$-order product formula (see Eq. \ref{eq:Suzuki-Trotter} and Eq. \ref{eq:general_product_formula_segment}) containing $r\sum_{\alpha=1}^{\nu_q}|\mathfrak X_\alpha|$ nonlocal gates by RUS, choosing
$K=\lceil\log_2(r\sum_{\alpha=1}^{\nu_q}|\mathfrak X_\alpha|/\delta)\rceil$ guarantees that all gate implementations succeed with probability at least $1-\delta$.

The entangled state $\ket{\phi_\theta}$ has entropy of $E(\phi_\theta)=h_2(\sin^2\theta)$.
For bipartite pure states, this quantity also gives the asymptotic entanglement cost: 
a block of many identical copies of $\ket{\phi_\theta}$ can be prepared from Bell pairs by entanglement dilution at a rate approaching $E(\phi_\theta)$ ebits per copy, with vanishing preparation error. 
The repeated structure of Trotterization provides precisely such blocks, because the same weak rotations recur across successive Trotter steps. 
It is therefore unnecessary to prepare each resource state exactly on a single-copy basis;
finite-accuracy dilution suffices, with the preparation error chosen within the overall simulation error budget. 
A finite-blocklength treatment and the resulting error accounting are given in Appendix~\ref{app:finite_accuracy_dilution}.
For the standard unmerged product formula implementation considered here, we further take, without loss of generality, each nonlocal sub-layer to contain the same set of cross-cut interactions, $\mathfrak X_\alpha=\mathfrak X$, so that each segment contains $\nu_q|\mathfrak X|$ nonlocal Pauli rotations. 
Any cancellations, commuting-group optimizations or mergers between adjacent segments can only reduce this number and therefore leave the upper bounds derived below unchanged.

Having established the entanglement cost of a single distributed gate $U_\theta$, we now extend the analysis to a complete step $S_q(\tau)$, whose total entanglement cost is upper bounded by
\begin{align}
    \sum_{\alpha=1}^{\nu_q}
    \sum_{x\in\mathfrak{X}_\alpha}
    9\tau |c_{\alpha,x} J_x|
    =
    &\,
    9\tau\sum_{x\in\mathfrak{X}}
    \left(\sum_{\alpha=1}^{\nu_q}|c_{\alpha,x}|\right) |J_x|\notag\\
    \leqslant
    &\, 9\tau\nu_q\sum_{x\in\mathfrak{X}}|J_x|,
\end{align}
where the inequality follows from $|c_{\alpha,x}|\leqslant 1$ for all $\alpha$ and $x$.
The full distributed simulation comprises $r$ times of $S_q(\tau)$ (see Eq.~\eqref{eq:Suzuki-Trotter}), giving the following bound on the total entanglement cost of the RUS distributed simulation, $C_{\mathrm{RUS}}$:
\begin{align}\label{eq:RUS_Upper}
    C_{\mathrm{RUS}}
    \leqslant
    r\left(9\tau\nu_q\sum_{x\in\mathfrak{X}}|J_x|\right)
    =  \left(9\nu_q\sum_{x\in\mathfrak{X}}|J_x|\right)t.
\end{align}
The dependence on the Trotter number $r$ therefore disappears entirely: 
the required entanglement grows at most linearly with the simulated evolution time $t$.
This yields our main result.

\begin{theorem}[{\bf Entanglement Cost of DQS}]
\label{thm:main}
Consider the bipartite Hamiltonian $H=H_A+H_B+H_{AB}$, with $H_{AB}$ defined in Eq.~\eqref{eq:H_AB}.
For any fixed-order product formula simulation of $e^{-\ii Ht}$ in Eq.~\eqref{eq:Suzuki-Trotter}, with each nonlocal Pauli rotation implemented by the RUS protocol, its entanglement cost grows at most linearly with the evolution time; that is
\begin{align}\label{eq:Upper}
    C_{\mathrm{RUS}}=\mathcal{O}(t\sum_{x\in\mathfrak{X}}|J_x|).
\end{align}
\end{theorem}

This scaling contrasts sharply with naive teleportation-based distributed simulation, where every nonlocal gate incurs a fixed entanglement cost and the total resource requirement grows with the number of Trotter steps $r$. 
As the Trotterization is refined, this fixed per-gate cost accumulates into a substantial overhead, revealing a central advantage of our RUS-based approach. 

Known communication-complexity lower bounds show that there exist bipartite
Hamiltonian simulation instances for which simulating \(e^{-iHt}\) to constant
error requires \(\Omega(t)\) qubits of quantum communication across the
bipartition~\cite{feng2024distributedquantumsimulation}.
This result follows from a reduction to the \textsc{inner product} (Boolean function) problem in the regime $\| H\|t \leqslant n$, where $n$ is the number of qubits available to a single quantum processing unit.
The communication bound in turn translates directly into a lower bound on the required entanglement:
a distributed simulation protocol consuming $n$ ebits can be recast as a quantum communication protocol by locally preparing $n$ maximally entangled pairs and transmitting one half of each pair across the bipartition using $n$ noiseless qubit channels.
The entanglement cost is therefore lower-bounded by
$
 \Omega(t).
$
Together with the upper bound in Eq.~\ref{eq:Upper}, this establishes the optimal linear scaling with evolution time.

\begin{corollary}[{\bf Optimal Time-Scaling for DQS}]
\label{cor:main}
There exists a distributed quantum simulation case for which the $q$th-order product formula in Eq.~\eqref{eq:Suzuki-Trotter} achieves the optimal linear dependence of entanglement cost on evolution time,
\begin{align}\label{eq:Optimal}
    C_{\mathrm{RUS}}=\Theta(t).
\end{align}
This scaling is attained when all nonlocal Pauli rotations are implemented using the RUS protocol.
\end{corollary}

\begin{figure}[t]
    \centering
    \includegraphics[width=\linewidth]{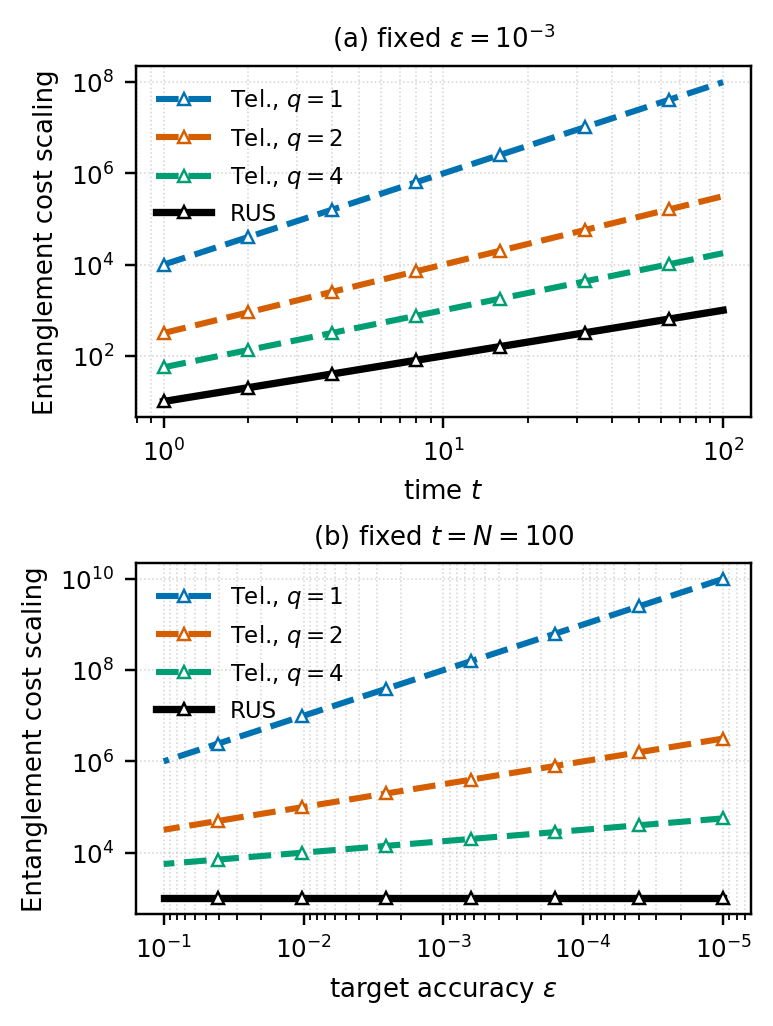}
    \caption{\textbf{Upper-bound scaling of the entanglement cost for distributed Trotterization.}
    The plotted quantity is an upper-bound scaling of entanglement cost for simulating a 2-dimensional lattice with $N=100$. Fixed prefactors, including $\nu_q$, the commutator factor $\alpha_q$, and other constants independent of the plotted variables, are omitted.
    (a) At fixed target error $\epsilon$, the cost of standard teleportation-based implementation grows superlinearly with evolution time $t$, whereas the RUS cost remains linear in $t$.
    (b) At fixed $t=N=100$, the cost of standard teleportation-based implementation increases as $\epsilon^{-1/q}$ with increasing simulation accuracy, whereas the RUS cost is independent of $\epsilon$.
    Dashed curves correspond to teleportation-based distributed simulation for product-formula orders $q=1,2,4$; 
    the solid black line denotes RUS.}
    \label{fig:teleport_vs_rus}
\end{figure}

\noindent \textbf{Resource Benchmark}---A direct comparison between the RUS and teleportation-based approaches reveals a fundamental difference in their resource scaling, and a further advantage of RUS emerges once simulation accuracy is taken into account.
Approximating the target evolution $e^{-iHt}$ with the $q$th-order product formula in Eq.~\eqref{eq:Suzuki-Trotter} incurs a simulation error $\epsilon$ that scales as $\epsilon = \mathcal{O}(\alpha_{q}\,t^{q+1}/r^q)$, where $\alpha_{q}$ denotes the relevant sum of \((q+1)\)-fold nested
commutator norms~\cite{childsTheoryTrotterError2021}.
A sufficient choice to achieve target accuracy $\epsilon$ is
\begin{align}
    r = \mathcal{O}(\alpha_{q}^{1/q}\,\epsilon^{-1/q}\,t^{1+1/q})
\end{align}
Trotter steps.

For standard teleportation-based simulation, each nonlocal gate carries a fixed cost of 2 ebits.  With $r\nu_q|\mathfrak{X}|$ such gates required to realize the full evolution, the total entanglement cost $C_{\mathrm{TEL}}^{(q)}$ becomes
\begin{align}\label{eq:q_TEL}
    C_{\mathrm{TEL}}^{(q)}=
    \mathcal{O}(\alpha_{q}^{1/q}\,\epsilon^{-1/q}\nu_q|\mathfrak{X}|\,t^{1+1/q}).
\end{align}
By contrast, the RUS entanglement cost $C_{\mathrm{RUS}}^{(q)}$ is bounded by Eq.~\eqref{eq:RUS_Upper}.
Defining $|J_x|=\mathcal{O}(1)$, this immediately leads to
\begin{align}\label{eq:q_RUS}
    C_{\mathrm{RUS}}^{(q)}=\mathcal{O}(\nu_q|\mathfrak{X}|\,t).
\end{align}
At fixed order $q$, the teleportation cost grows as $\epsilon^{-1/q}$ as the target error is reduced, whereas the RUS cost remains independent of $\epsilon$.
The resource advantage of RUS therefore grows with the required simulation accuracy.

Consider bounded-strength Hamiltonian on a two-dimensional $N^{1/2}\times N^{1/2}$ square lattice, with $|J_x|=\mathcal{O}(1)$.
A bipartition cuts $|\mathfrak{X}|=\mathcal{O}(N^{1/2})$ interaction terms.
For simplicity, we suppress the fixed $q$-dependent factors $\alpha_{q}$ and $\nu_q$, so that the two costs scale as $C_{\mathrm{TEL}}^{(q)}=\mathcal{O}(\epsilon^{-1/q}N^{1/2}t^{1+1/q})$ (see Eq.~\eqref{eq:q_TEL}) and $C_{\mathrm{RUS}}^{(q)}=\mathcal{O}(N^{1/2}t)$ (see Eq.~\eqref{eq:q_RUS}).
Figure~\ref{fig:teleport_vs_rus} makes this separation clear for $N=100$.
At fixed accuracy, the teleportation cost grows superlinearly with evolution time, whereas the RUS cost remains linear in $t$, causing the resource gap to widen for longer simulations (see Fig.~\ref{fig:teleport_vs_rus}(a)).
At fixed $t$, increasing the required accuracy leads to a second separation: 
the teleportation cost grows as $\epsilon^{-1/q}$ and diverges as $\epsilon\to0$, while the RUS cost remains constant (see Fig.~\ref{fig:teleport_vs_rus}(b)).
Hence, RUS achieves the same simulation task with progressively less entanglement than teleportation as either the evolution time or the required accuracy increases.

\noindent \textbf{Multipartite Generalization}---The analysis so far has focused on bipartite interactions. 
A broader setting arises when a nonlocal rotation acts jointly across $m$ spatially separated parties,
$
    V_\theta 
    \coloneqq
    e^{-\ii\theta P_1\otimes P_2\otimes\cdots\otimes P_m},
 $
where $P_j^2=\id$ acts locally on party $j$.
In this setting, the weak Bell pair is replaced by the GHZ-type entangled state $\ket{\mathrm{GHZ}_\theta}\coloneqq\cos\theta\ket{0}^{\otimes m}  - \ii\sin\theta\ket{1}^{\otimes m}$.
Each party uses its share of this state to control the corresponding local operation $P_j$, after which the resource qubits are measured in the Hadamard basis.
For an outcome string $x\in\{0,1\}^m$, the resulting operation on the data system is
\begin{align}
    \left(\cos\theta\right)\id - 
    \ii\left((-1)^{|x|} \sin\theta\right) P_1\otimes\cdots\otimes P_m.
\end{align}
The parity of $x$ therefore fixes the sign of the rotation: 
even parity realizes $V_\theta$, whereas odd parity yields $V_\theta^\dagger$.

An odd-parity outcome can be corrected by the same angle-doubling procedure used in bipartite RUS, giving a success probability $1-2^{-K}$ after $K$ rounds. 
The RUS mechanism therefore carries over from bipartite to genuinely multipartite nonlocal rotations without changing its basic structure.
The associated resource cost, however, becomes network dependent, as preparing and distributing multipartite entangled states $\ket{\mathrm{GHZ}_\theta}$ depends on both the network architecture and the chosen measure of multipartite entanglement. 
A full characterization of this cost lies beyond the bipartite analysis developed here and is left for future work.

\noindent \textbf{Discussion}---Reducing Trotter error is usually understood to come at a resource cost~\cite{zhaoHamiltonianSimulationRandom2021,zeng2025simple}: 
finer time steps, higher-order formulas or other variants introduce more nonlocal gates.
In distributed simulation, this has been assumed to require proportionally more entanglement, because teleportation assigns the same cost to every nonlocal operation. 
Our results show that this link is not fundamental. 
The relevant resource can instead track the physical interaction accumulated during the evolution, rather than the number of gates introduced by a particular decomposition. 
RUS realizes this principle by matching the entanglement of the resource state to the rotation strength: 
as the Trotterization is refined, more nonlocal rotations appear, but each becomes correspondingly weaker, leaving the total expected cost unchanged. 
Distributed Trotterization can therefore be refined to arbitrarily high accuracy without increasing its leading entanglement requirement. 
For any fixed product formula order, this requirement remains linear in evolution time and independent of the Trotter error, matching the lower bound set by quantum communication complexity.

For multipartite distributed quantum simulation, the present construction uses weak GHZ-type resources shared among the participating nodes. 
The remaining question is how to quantify the cost of preparing these resources from pairwise ebits on a given network, where the answer can depend on the network topology and on the routing structure connecting the nodes.
Establishing this topology-dependent entanglement cost would provide an operational resource theory for multipartite distributed simulation.
Further savings may be possible by exploiting structure within a simulation layer. 
Instead of implementing interaction terms independently, global protocols may make use of commutativity or cancellations among them, potentially reducing the constant prefactor in the entanglement cost \cite{Xiaomingsun11575660}.
More broadly, it remains open whether the same entanglement resource scaling extends beyond product formulas to simulation methods, including  linear combinations of unitaries \cite{berrySimulatingHamiltonianDynamics2015, childsHamiltonianSimulationUsing2012, longDualityQuantumComputing2011,sun2026high}, quantum signal processing and quantum singular value transformation \cite{lowHamiltonianSimulationQubitization2019,lowOptimalHamiltonianSimulation2017, Gily_n_2019}. Beyond entanglement, extending the analysis to fault-tolerant setting, where non-Clifford resources such as magic states and their associated implementation overheads become relevant, offers a further direction towards a complete resource picture of distributed quantum simulation \cite{yuan2026distributedresourcetheoryentanglement}.

\begin{acknowledgments}
T. F. thanks Zhong-Xia Shang, Xiangyi Meng, Wenjun Yu, Chenfeng Cao and Jue Xu for their helpful discussion. 
Q. Z. acknowledges funding from Quantum Science and Technology-National Science and Technology Major Project 2024ZD0301900, National Natural Science Foundation of China (NSFC) via Project No. 12347104 and No. 12305030, Hong Kong Research Grant Council (RGC) via No. 27300823, 17310926, N\_HKU718/23, and R6010-23. T. F. and Q. Z. thank the support of the CCF-Tencent Rhino-Bird Open Research Fund (No. T105-TEG-2024080200001). Y. X. is supported by A*STAR under its Career Development Fund (C243512002).
\end{acknowledgments}

\bibliography{ref_new,ref,refer,ref_aps}

\newpage
\begin{widetext}
    
\appendix

\newpage

\section{Repeat-until-success protocol for nonlocal rotation gates}

The complete procedure for the Repeat-until-success (RUS) protocol to implement a nonlocal rotation gate $U_\theta =e^{-i\theta P_A\otimes P_B }$ is summarized in Alg.~\ref{alg:RUS_Appendix}. 

\begin{algorithm}[H]
\caption{RUS Protocol}
\label{alg:RUS_Appendix}
\begin{algorithmic}[1]

\Require Data state $\ket{\psi}_{AB}$; 
Pauli strings $P_A$ and $P_B$;
target angle $\theta$; 
desired success probability threshold $1-\delta\in(0,1)$
\Ensure $U_\theta\ket{\psi}$ upon success, with
$\Pr[\textsc{success}]\geqslant 1-\delta$

\State $K
\gets
\left\lceil
\log_2\!\left(\frac{1}{\delta}\right)
\right\rceil$

\For{$j=0,1,\ldots,K-1$}
    \State Set $\alpha_j \gets 2^j\theta$

    \State Alice and Bob prepare the entangled state $\ket{\phi_{\alpha_j}}\coloneqq\cos\alpha_j\ket{00}-\ii\sin\alpha_j\ket{11}$

    \State Alice and Bob apply controlled-$P_A$ and controlled-$P_B$,
    respectively, using $\ket{\phi_{\alpha_j}}$ as controls

    \State Measure $\ket{\phi_{\alpha_j}}$ in the Hadamard basis, obtaining $m_A, m_B\in\{0,1\}$

    \If{$m_A\oplus m_B=0$}
        \State \Return \textsc{success}
        \Comment{the evolution is $U_\theta$}
    \Else
        \State Continue to the next round
        \Comment{the accumulated evolution is
        $U_{-(2^{j+1}-1)\theta}$}
    \EndIf
\EndFor

\State \Return \textsc{failure}
\Comment{$U_{-(2^{K}-1)\theta}$,
with probability $2^{-K}\leqslant \delta$}

\end{algorithmic}
\end{algorithm}

\section {Finite-time RUS cost of Trotterization}
\label{app:longtime_rus_cost}

In this appendix we prove the linear-in-time cost bound used in
Theorem~\ref{thm:main}. The only estimate needed for the
Trotterization cost is an upper bound on the expected RUS cost of one weak
rotation. Recall that
\begin{equation}
    C_K(\theta)
    =
    \sum_{j=0}^{K-1}2^{-j}
    h_2\!\left(\sin^2(2^j\theta)\right),
\end{equation}
and define
$
    C_\infty(\theta)
    =
    \sum_{j=0}^{\infty}2^{-j}
    h_2\!\left(\sin^2(2^j\theta)\right).
$

\begin{lemma}[Lemma 1 in the main text]
\label{lem:Cinfty_upper}
For every $K\ge1$ and every $0<|\theta|\le1$,
\begin{equation}
    C_K(\theta)
    \le
    C_\infty(\theta)
    \le
    9|\theta|.
    \label{eq:CK_upper_9}
\end{equation}
\end{lemma}

\begin{proof}
Let $f(x)=h_2(\sin^2 x)$. We use the elementary
entropy bound
$
    h_2(p)
    \le
    p\log_2\frac{e}{p}
    $ with $ 0<p\le1,$
which follows from $-(1-p)\ln(1-p)\le p$. Hence, for $|x|\le1$,
\begin{equation}
    f(x)
    =
    h_2(\sin^2 x)
    \le
    x^2\log_2\frac{e}{x^2}.
    \label{eq:f_bound_simple}
\end{equation}

Choose $J\ge0$ such that
$ 2^J|\theta|\le1<2^{J+1}|\theta|,$
and write $b=2^J|\theta|$, so that $b\in(1/2,1]$. We split the infinite sum at
$j=J$. For $j\le J$, set $n=J-j$. Then $2^j|\theta|=b2^{-n}$, and
Eq.~\eqref{eq:f_bound_simple} gives
\begin{align}
    \sum_{j=0}^{J}2^{-j}f(2^j\theta)
    &\le
    |\theta|
    \sum_{n=0}^{J}
    b2^{-n}
    \left(
        \log_2\frac{e}{b^2}
        +2n
    \right)
    \nonumber \le
    |\theta| b
    \left[
        2\log_2\frac{e}{b^2}
        +4
    \right].
    \label{eq:small_angle_sum_bound}
\end{align}
Here we used $\sum_{n\ge0}2^{-n}=2$ and
$\sum_{n\ge0}n2^{-n}=2$. Define function
$
    \Phi(b)
    =
    b\left[
        2\log_2\frac{e}{b^2}
        +4
    \right]$ and one can verify that$ \Phi(b)$
is increasing on $b\in(1/2,1]$ since
$ \Phi'(b)
    >
    0.$
Therefore $\Phi(b)\le \Phi(1)=2\log_2 e+4$, and
\begin{equation}
    \sum_{j=0}^{J}2^{-j}f(2^j\theta)
    \le
    (2\log_2 e+4)|\theta|.
    \label{eq:small_angle_sum_final}
\end{equation}

For the tail $j>J$, we use $0\le f(x)\le1$:
\begin{equation}
    \sum_{j=J+1}^{\infty}2^{-j}f(2^j\theta)
    \le
    \sum_{j=J+1}^{\infty}2^{-j}
    =
    2^{-J}.
\end{equation}
Since $b=2^J|\theta|>1/2$, we have $2^{-J}=|\theta|/b<2|\theta|$. Combining this with
Eq.~\eqref{eq:small_angle_sum_final}, we obtain
\begin{equation}
    C_\infty(\theta)
    \le
    (2\log_2 e+6)|\theta|
    <
    9|\theta|.
\end{equation}
Finally, $C_K(\theta)\le C_\infty(\theta)$ because all summands are
nonnegative. This proves the claim.
\end{proof}

\subsection{Cost of Trotterization with the RUS protocol}
\label{subsec:rus_trotter_cost}

We now apply the single-gate RUS cost bound to a full product-formula
simulation.  We state the result for a general fixed-order product formula,
of which first-order Trotterization is a special case.

Let $q$denote the product-formula order and let $\tau=t/r$.  In one
Trotter segment, the nonlocal part of the formula can be written
schematically as a sequence of nonlocal Pauli rotations \cite{childsTheoryTrotterError2021}
\begin{equation}
    \exp[-i\theta_{\alpha,x}P_x],
    \qquad
    \theta_{\alpha,x}
    =
    a_{\alpha,x}J_x\tau
    =
    a_{\alpha,x}J_x\frac{t}{r},
    \label{eq:general_trotter_angle}
\end{equation}
where $\alpha=1,\ldots,\nu_q$ labels the nonlocal sweeps in one segment,
$\mathfrak{X}_\alpha$ is the set of nonlocal terms appearing
in sweep $\alpha$, and $a_{\alpha,x}$ are fixed formula-dependent
coefficients. For fixed product-formula order $q$, the number of sweeps
$\nu_q$ and the coefficients $a_{\alpha,x}$ are independent of $t$,
$r$, and the target precision.

For each interaction term $x$, define its total absolute coefficient weight
within one segment by
$
    A_{q,x}
    =
    \sum_{\alpha:\,x\in\mathfrak{X}_\alpha}
    |a_{\alpha,x}|.
$
We also write $A_q=\max_{x\in\mathfrak X}A_{q,x}$.  For first-order
Trotterization, $\nu_q=1$, $\mathfrak{X}_1=\mathfrak{X},$ $
a_{1,x}=1$, and hence $A_{q,x}=1$. For the standard, unmerged product-formula implementation considered here, each
nonlocal sublayer contains the same set of cross-cut interactions,
$\mathfrak X_\alpha=\mathfrak X$. Thus one segment contains
$\nu_q|\mathfrak X|$ nonlocal Pauli rotations. Possible cancellations,
commuting-group optimizations, or mergers between adjacent segments can only
reduce this number and do not affect the upper bounds derived below.

For one nonlocal rotation with angle $\theta_{\alpha,x}$, the expected cost
of a \(K\)-round RUS attempt is
\begin{equation}
    C_K(\theta_{\alpha,x})
    =
    \sum_{j=0}^{K-1}
    2^{-j}
    h_2\!\left(
        \sin^2(2^j\theta_{\alpha,x})
    \right).
    \label{eq:CK_general_angle}
\end{equation}
Therefore the total expected ebit cost over all $r$ segments is
\begin{equation}
    C_{\rm tot}^{(K)}
    =
    r
    \sum_{\alpha=1}^{\nu_q}
    \sum_{x\in\mathfrak{X}_\alpha}
    C_K(\theta_{\alpha,x}).
    \label{eq:Ctot_general_q_xxact}
\end{equation}

Each $K$-round RUS attempt fails only if all $K$ heralded signs are
unfavorable, which occurs with probability $2^{-K}$.  Indeed, for a
correction angle $\alpha'$, the measurement outcome $x\in\{0,1\}^{k_x}$
induces the Kraus operator
\begin{equation}
    M_x
    =
    2^{-k_x/2}
    \left(
        \cos\alpha'\,I
        -
        i(-1)^{|x|}\sin\alpha'\,P_x
    \right),
\end{equation}
and hence $M_x^\dagger M_x=2^{-k_x}I$.  Thus each parity occurs with
probability $1/2$, independently of the input state and independently of the
angle.  The conclusion does not require the correction angles
$2^j\theta_{\alpha,x}$ to remain small.

The total number of nonlocal RUS attempts in the circuit is
$
    M
    =
    r
    \sum_{\alpha=1}^{\nu_q}
    |\mathfrak{X}_\alpha|.
$
By the union bound,
$
    P_{\rm succ}
    \ge
    1-M2^{-K}.
$
Thus choosing
$
    K
    =
    \left\lceil
        \log_2\frac{M}{\delta}
    \right\rceil
$
ensures \(P_{\rm succ}\ge1-\delta\).

Assume now that all elementary Trotter angles satisfy
$ |\theta_{\alpha,x}|\le1$ 
for all  $\alpha,x.$
By Lemma~\ref{lem:Cinfty_upper},
$  C_K(\theta_{\alpha,x})
    \le
    9|\theta_{\alpha,x}|.$
Substituting Eq.~\eqref{eq:general_trotter_angle} into
Eq.~\eqref{eq:Ctot_general_q_xxact}, we obtain

\begin{align}
    C_{\rm tot}^{(K)}
    &\le
    9r
    \sum_{\alpha=1}^{\nu_q}
    \sum_{x\in\mathfrak{X}_\alpha}
    |a_{\alpha,x}J_x|\frac{t}{r}
    \nonumber=
    9t
    \sum_{x\in\mathfrak{X}}
    A_{q,x}|J_x|
    \nonumber\le
    9A_q t
    \sum_{x\in\mathfrak{X}}
    |J_x|.
    \label{eq:linear_time_general_q}
\end{align}
Therefore, for every fixed product-formula order $q$, the expected RUS cost
is linear in the simulated time and independent of the Trotter segment number
$r$, up to formula-dependent constants.
\begin{theorem}[Equivalent to Theorem 1 in the main text]
\label{thm:longtime_rus_cost}
Consider an $r$-segment product-formula simulation of fixed order $q $in
the bipartite setting.  Let the nonlocal rotations in one segment be
$
    \exp[-i\theta_{\alpha,x}P_x],
   $ with $
    \theta_{\alpha,x}
    =
    a_{\alpha,x}J_x t/r,
$
with \(\alpha=1,\ldots,\nu_q\) and \(x\in\mathfrak{X}_\alpha\).  Define
$
    M
    =
    r\sum_{\alpha=1}^{\nu_q}|\mathfrak{X}_\alpha|
$
and
$A_{q,x} =\sum_{\alpha:\,x\in\mathfrak{X}_\alpha}|a_{\alpha,x}|,
   $ $
    A_q=\max_{x\in\mathfrak{X}}A_{q,x}.
$
A $K$-round RUS implementation has expected ebit cost
\begin{equation}
    C_{\rm tot}^{(K)}
    =
    r
    \sum_{\alpha=1}^{\nu_q}
    \sum_{x\in\mathfrak{X}_\alpha}
    \sum_{j=0}^{K-1}
    2^{-j}
    h_2\!\left(
        \sin^2(2^j\theta_{\alpha,x})
    \right),
    \label{eq:theorem_general_cost}
\end{equation}
and success probability
$
    P_{\rm succ}\ge1-M2^{-K}.
$
Thus $K=\lceil\log_2(M/\delta)\rceil$ gives
$P_{\rm succ}\ge1-\delta$.  If all elementary angles satisfy
$|\theta_{\alpha,x}|\le1$, then
$
    C_{\rm tot}^{(K)}
    \le
    9A_qt
    \sum_{x\in\mathfrak{X}}|J_x|.
$
In particular, for fixed $q$,
\begin{equation}
    C_{\rm tot}^{(K)}
    =
    O\!\left(
        t\sum_{x\in\mathfrak{X}}|J_x|
    \right),
\end{equation}
where the hidden constant depends only on the chosen product formula.
\end{theorem}
A coarser but simpler bound is obtained from $A_{q,x}\le\nu_q$ whenever
$|a_{\alpha,x}|\le1$. In that case
$
    C_{\rm tot}^{(K)}
    \le
    9\nu_q t\sum_{x\in\mathfrak{X}}|J_x|.
$

\subsection{Protocol-specific lower bound for the RUS cost}
\label{app:rus_protocol_lower_bound}

In the main text we only use the upper bound
$C_K(\theta)\le 9|\theta|$ to prove the linear-in-time simulation cost.
For completeness, we record a matching lower bound for the proposed RUS
protocol when the cutoff reaches a constant-angle correction scale. This is a
protocol-specific statement and should not be interpreted as a lower bound on
all possible LOCC implementations.

\begin{lemma}[Protocol-specific saturation]
\label{lem:CK_protocol_lower}
Fix a constant $a\in(0,\pi/4)$. Define
$ m_a= \min_{x\in[a,2a]} h_2(\sin^2 x).
$
Then $m_a>0$. For every $0<|\theta|\le a$, if
$
    K\ge
    \left\lceil
    \log_2\frac{a}{|\theta|}
    \right\rceil+1,$
then
$
    C_K(\theta)
    \ge
    \frac{m_a}{2a}|\theta|.$
Together with $C_K(\theta)\le9|\theta|$, this implies
\begin{equation}
    C_K(\theta)=\Theta(|\theta|)
\end{equation}
for the proposed RUS protocol whenever the cutoff reaches the constant-angle
scale.
\end{lemma}

\begin{proof}
Let
$
    j_*=
    \left\lceil
    \log_2\frac{a}{|\theta|}
    \right\rceil.$
Then
$
    a\le 2^{j_*}|\theta|<2a.
$
Because $a\in(0,\pi/4)$, the interval $[a,2a]$ is contained in
$(0,\pi/2)$. Hence $h_2(\sin^2 x)$ is continuous and strictly positive on
$[a,2a]$, so $m_a>0$.

If $K\ge j_*+1$, the term $j=j_*$ is included in $C_K(\theta)$. Therefore
\begin{align}
    C_K(\theta)=
    \sum_{j=0}^{K-1}
    2^{-j}h_2\!\left(\sin^2(2^j\theta)\right)
    \ge
    2^{-j_*}h_2\!\left(\sin^2(2^{j_*}\theta)\right)
    \ge
    2^{-j_*}m_a.
\end{align}
Using $2^{j_*}|\theta|<2a$, we obtain
$
    2^{-j_*}>
    \frac{|\theta|}{2a}.
$
Thus
\begin{equation}
    C_K(\theta)
    \ge
    \frac{m_a}{2a}|\theta|.
\end{equation}
The upper bound
$C_K(\theta)\le9|\theta|$ then gives the claimed $\Theta(|\theta|)$ scaling
for this RUS protocol.
\end{proof}

\section{Finite-Accuracy Preparation of Weak Entangled Resources}
\label{app:finite_accuracy_dilution}

In the main text we quote the entanglement cost of a weak Bell resource
\begin{equation}
    \ket{\phi_\alpha}
    =
    \cos\alpha\ket{00}
    -
    i\sin\alpha\ket{11}
\end{equation}
as its entropy of entanglement,
$h_2(\sin^2\alpha)$. This is the asymptotic Bell-pair rate. In an actual
finite implementation, however, the resource states need not be prepared
perfectly. It is sufficient to prepare them to an accuracy comparable to the
error budget already assigned to the corresponding Trotter layer. This
appendix makes this statement precise.

For states, we use the trace distance
$   D(\rho,\sigma)
    =
    \frac12\|\rho-\sigma\|_1.$
For channels, we use the corresponding operational distance
\begin{equation}
    D_\diamond(\mathcal N,\mathcal M)
    =
    \frac12\|\mathcal N-\mathcal M\|_\diamond .
\end{equation}

\subsection{Stability of the RUS channel under resource errors}

The first observation is that approximate resource states translate directly
into approximate implemented channels.

\begin{lemma}[Resource error implies channel error]
\label{lem:resource_to_channel_xrror}
Let $\mathcal R$ be any fixed LOCC procedure that consumes an ideal resource
state $\omega$ and implements a channel $\mathcal N_\omega$ on the data.
If the resource state is replaced by $\widetilde\omega$, producing the
channel $\mathcal N_{\widetilde\omega}$, then
\begin{equation}
    D_\diamond(\mathcal N_\omega,\mathcal N_{\widetilde\omega})
    \le
    D(\omega,\widetilde\omega).
    \label{eq:resource_channel_stability}
\end{equation}
The same bound holds if $\mathcal R$ includes classical outcome registers,
success/failure flags, or adaptive measurements.
\end{lemma}

\begin{proof}
The LOCC procedure, including all measurements and classical registers, is a
CPTP map acting jointly on the input data and the resource state. For any
input state $\rho_{RA}$, including an arbitrary reference system $R$, the
two output states differ only through the replacement
\(\omega\mapsto\widetilde\omega\). Contractivity of trace distance under CPTP
maps gives
\begin{equation}
    D\!\left[
        \mathcal R(\rho_{RA}\otimes\omega),
        \mathcal R(\rho_{RA}\otimes\widetilde\omega)
    \right]
    \le
    D(\omega,\widetilde\omega).
\end{equation}
Taking the supremum over $\rho_{RA}$ gives
Eq.~\eqref{eq:resource_channel_stability}.
\end{proof}

Thus, if the approximate resource package used for one nonlocal RUS attempt
is within trace distance $\delta$ of the ideal resource package, then the
implemented RUS channel differs from the ideal-resource RUS channel by at
most $\delta$ in diamond distance.

\subsection{Finite weak entangled-pair preparation by finite-accuracy dilution}

Before stating the finite-copy preparation bound, we spell out the
Schmidt-basis structure of $n$ weak Bell resources.   Let
\begin{equation}
    \ket{\phi_\alpha}
    =
    \cos\alpha\ket{0}_A\ket{0}_B
    -
    i\sin\alpha\ket{1}_A\ket{1}_B .
\end{equation}
Consider $n$ copies of this state, with copy index
$\ell=1,\ldots,n$.  We write
$A^n=A_1\cdots A_n$ and $B^n=B_1\cdots B_n$.  For a bit string
$x=(x_1,\ldots,x_n)\in\{0,1\}^n$, define
\[
    \ket{x}_{A^n}
    =
    \ket{x_1}_{A_1}\cdots\ket{x_n}_{A_n},
    \qquad
    \ket{x}_{B^n}
    =
    \ket{x_1}_{B_1}\cdots\ket{x_n}_{B_n},
\]
and let $|x|=\sum_{\ell=1}^n x_\ell$ be its Hamming weight.  Expanding the
tensor product gives
\begin{equation}
    \ket{\phi_\alpha}^{\otimes n}
    =
    \sum_{x\in\{0,1\}^n}
    (\cos\alpha)^{n-|x|}
    (-i\sin\alpha)^{|x|}
    \ket{x}_{A^n}\ket{x}_{B^n}.
    \label{eq:n_copy_weak_bell_general}
\end{equation}
Equivalently, when $0\le\alpha\le\pi/2$, writing
$p=\sin^2\alpha$ and $1-p=\cos^2\alpha$, this becomes
\begin{equation}
    \ket{\phi_\alpha}^{\otimes n}
    =
    \sum_{x\in\{0,1\}^n}
    (1-p)^{(n-|x|)/2}
    (-i)^{|x|}
    p^{|x|/2}
    \ket{x}_{A^n}\ket{x}_{B^n}.
    \label{eq:n_copy_weak_bell_pq}
\end{equation}
For general $\alpha$, the same construction applies after replacing
$\cos\alpha$ and $\sin\alpha$ by their absolute values in the Schmidt
coefficients; the signs and phases are known and can be incorporated by local
phase conventions.
Consequently the Hamming weight $|x|$ is distributed
as a binomial random variable $X\sim{\rm Bin}(n,p)$, since there are
$\binom{n}{w}$ strings with $|x|=w$.  Therefore, the total Schmidt weight
of all strings with $|x|\le L$ is
\begin{equation}
    F_{n,L}(\alpha)
    =
    \sum_{|x|\le L}(1-p)^{n-|x|}p^{|x|}
    =
    \sum_{w=0}^{L}\binom{n}{w}(1-p)^{n-w}p^w
    =
    \Pr[{\rm Bin}(n,p)\le L].
    \label{eq:FnL_binomial}
\end{equation}
The discarded Schmidt weight is correspondingly
\begin{equation}
    1-F_{n,L}(\alpha)
    =
    \Pr[{\rm Bin}(n,p)>L].
    \label{eq:discarded_binomial_tail}
\end{equation}
The finite-copy dilution protocol below simply keeps the low-weight Schmidt
subspace $|x|\le L$ and discards the binomial tail.

\begin{proposition}[Truncated finite-copy dilution]
\label{thm:finite_dilution_truncation}
Fix integers $n\ge1$ and $0\le L\le n$. Define
$D_L
    =
    \sum_{w=0}^{L}\binom{n}{w},
$ $m_L
    =
    \left\lceil \log_2 D_L \right\rceil ,
$
and
$
    F_{n,L}(\alpha)
    =
    \sum_{w=0}^{L}
    \binom{n}{w}
    (1-p)^{n-w}p^w .
$
Using $m_L$ Bell pairs, Alice and Bob can prepare by LOCC a pure state $\ket{\Phi_{n,L}}$ satisfying
\begin{equation}
    D\!\left(
        \ket{\Phi_{n,L}}\bra{\Phi_{n,L}},
        \ket{\phi_\alpha}\bra{\phi_\alpha}^{\otimes n}
    \right)
    \le
    \sqrt{1-F_{n,L}(\alpha)} .
    \label{eq:block_trace_distance_bound}
\end{equation}
Equivalently, since
$
    1-F_{n,L}(\alpha)
    =
    \Pr[{\rm Bin}(n,\sin^2\alpha)>L],
$
block trace distance at most \(\delta\) is guaranteed whenever
$
    \Pr[{\rm Bin}(n,\sin^2\alpha)>L]\le \delta^2 .
$
\end{proposition}

\begin{proof}
Let $T_L=\{x\in\{0,1\}^n: |x|\le L\}$. The normalized truncated state is
\begin{equation}
    \ket{\Phi_{n,L}}
    =
    \frac{1}{\sqrt{F_{n,L}(\alpha)}}
    \sum_{x\in T_L}
    (1-p)^{(n-|x|)/2}
    (-i)^{|x|}
    p^{|x|/2}
    \ket{x}_A\ket{x}_B .
    \label{eq:truncated_state}
\end{equation}
Its Schmidt rank is at most $D_L=|T_L|$. Therefore it can be prepared from
$m_L=\lceil\log_2D_L\rceil$ Bell pairs: Alice and Bob first generate a
maximally entangled state of Schmidt rank at least $D_L$, and then convert
it by deterministic LOCC (with the majorization theorem) into the desired pure state with Schmidt rank
$D_L$ \cite{Nielsen_1999}.

The squared overlap with the ideal product state is exactly
\begin{equation}
    \left|
    \bra{\Phi_{n,L}}\phi_\alpha^{\otimes n}
    \right|^2
    =
    F_{n,L}(\alpha).
\end{equation}
For pure states, the trace distance is
\begin{equation}
    D(\ket{\psi}\bra{\psi},\ket{\varphi}\bra{\varphi})
    =
    \sqrt{1-|\langle {\psi|\varphi\rangle}|^2}.
\end{equation}
This gives Eq.~\eqref{eq:block_trace_distance_bound}.
\end{proof}

By monotonicity of trace distance under partial trace, the same block
guarantee implies that every individual copy marginal is also close to the
ideal weak Bell state:
\begin{equation}
    D(\rho_j,\ket{\phi_\alpha}\bra{\phi_\alpha})
    \le
    \sqrt{1-F_{n,L}(\alpha)} ,
    \label{eq:single_copy_marginal_bound}
\end{equation}
where $\rho_j$ is the marginal state of the $j$-th prepared copy.
However, the block guarantee in Eq.~\eqref{eq:block_trace_distance_bound} is
stronger and is the one used to control the error of a full circuit, since
the prepared copies may be correlated.

\subsection{Matching finite-resource preparation to the Trotter error budget}
\label{subsec:matching_resource_trotter_xrror}

We now explain how finite-accuracy preparation of the weak resource states is
incorporated into the full adaptive RUS Trotterization protocol. The important
point is that the error must be assigned to an entire $K$-round RUS attempt,
not only to the first weak resource. Indeed, if several unfavorable heralded
outcomes occur, the protocol uses resource states with doubled angles
$\theta,2\theta,\ldots,2^{K-1}\theta$. The preparation errors of all
resources that may be used in the adaptive correction sequence must therefore
be included in the error budget.

Let $g$ label a nonlocal RUS attempt in the product-formula circuit. For a
fixed-order product formula, one may take $g=(s,\alpha,x)$, where
$s=1,\ldots,r$ labels the Trotter segment, $\alpha$labels the nonlocal
sweep inside the segment, and $x\in\mathfrak{X}_\alpha$. The corresponding
initial rotation angle is denoted by $\theta_g$. A $K$-round ideal RUS
attempt may use the resource states with angles
$ \theta_g,\;2\theta_g,\;\ldots,\;2^{K-1}\theta_g .
$
For the purpose of error analysis, define the ideal resource package
\begin{equation}
    \Omega_{g,K}
    =
    \bigotimes_{j=0}^{K-1}
    \ket{\phi_{2^j\theta_g}}\bra{\phi_{2^j\theta_g}} .
    \label{eq:ideal_resource_package}
\end{equation}
This package notation is only an analytical device: operationally, the
resources can still be prepared and consumed on demand, and the entanglement
cost is counted only when a round is actually reached.

Suppose the actual prepared resource for round $j$ of attempt $g $is
$\widetilde \phi_{g,j}$, satisfying
\begin{equation}
    D\!\left(
        \widetilde \phi_{g,j},
        \ket{\phi_{2^j\theta_g}}\bra{\phi_{2^j\theta_g}}
    \right)
    \le
    \delta_{g,j},
    \label{eq:round_resource_xrror}
\end{equation}
where $D(\rho,\sigma)=\frac12\|\rho-\sigma\|_1$. Let
$
    \eta_g
    =
    \sum_{j=0}^{K-1}\delta_{g,j},
$ 
we have the following Lemma.

\begin{lemma}[Error of an approximate RUS resource package]
\label{lem:rus_package_xrror}
Let $\mathcal R_{g,K}$ be the ideal $K$-round RUS instrument for attempt
$g$, including the classical success/failure flag, and let
$\widetilde{\mathcal R}_{g,K}$ be the same adaptive LOCC protocol with the
approximate resources $\widetilde \phi_{g,j}$. Then
\begin{equation}
    D_\diamond(
        \mathcal R_{g,K},
        \widetilde{\mathcal R}_{g,K}
    )
    \le
    \eta_g ,
    \label{eq:rus_package_diamond_bound}
\end{equation}
where $D_\diamond(\mathcal N,\mathcal M)=\frac12\|\mathcal N-\mathcal M\|_\diamond$.
\end{lemma}

\begin{proof}
Consider first the mathematical protocol in which all $K$ resource registers
are supplied at the beginning, although the adaptive LOCC map only accesses
the $j$-th register if round $j$ is reached. This is equivalent to the
on-demand implementation for the input-output channel. The
entire adaptive RUS procedure is then a CPTP map acting on the data and the resource package. By contractivity of trace distance under CPTP maps, replacing the ideal package by the approximate package changes the resulting
instrument by at most the trace distance between the two packages.

It remains to bound the package distance. By a telescoping argument for
tensor products,
\begin{equation}
    D\!\left(
        \bigotimes_{j=0}^{K-1}\widetilde\phi_{g,j},
        \bigotimes_{j=0}^{K-1}
        \ket{\phi_{2^j\theta_g}}\bra{\phi_{2^j\theta_g}}
    \right)
    \le
    \sum_{j=0}^{K-1}\delta_{g,j}
    =
    \eta_g .
\end{equation}
Combining these two facts gives Eq.~\eqref{eq:rus_package_diamond_bound}.
\end{proof}

Let
$
    M
    =
    r\sum_{\alpha=1}^{\nu_q}|\mathfrak{X}_\alpha| = r \nu_q |\mathfrak{X}|
$
be the total number of nonlocal RUS attempts in the $r$-segment
fixed-order product formula. A hybrid argument over the $M$ attempts gives
a total resource-preparation error at most
$
    \sum_{g=1}^{M}\eta_g .
$
Therefore, if $\varepsilon_{\rm res}$ is the total error budget assigned to
finite resource preparation, it is sufficient to impose
\begin{equation}
    \eta_g
    =
    \sum_{j=0}^{K-1}\delta_{g,j}
    \le
    \frac{\varepsilon_{\rm res}}{M}
    \label{eq:eta_budget_per_attempt}
\end{equation}
for every RUS attempt $g$. A simple uniform allocation is
$
    \delta_{g,j}
    =
    \frac{\varepsilon_{\rm res}}{MK}
  $ for $
    j=0,\ldots,K-1.
$
More refined allocations are possible, but $\delta_{g,j}
    = \frac{\varepsilon_{\rm res}}{MK}$
is sufficient for a rigorous error bound.

We now relate this to the product-formula error. For a $q$-th order product
formula, write
\begin{equation}
    \epsilon_{\rm PF}(t,r)
    \le
    c_q\alpha_{q}\frac{t^{q+1}}{r^q},
    \label{eq:PF_xrror_budget}
\end{equation}
where $\alpha_{q}$ is the sum of
$(q+1)$-fold nested-commutator norms and $c_q$ is a formula-dependent
constant. One may choose the resource-preparation budget to be comparable to
the product-formula error, for example
\begin{equation}
    \varepsilon_{\rm res}
    =
    O(\epsilon_{\rm PF}(t,r)).
\end{equation}
Then $\delta_{g,j}
    = \frac{\varepsilon_{\rm res}}{MK}$ gives
\begin{equation}
    \delta_{g,j}
    =
    O\!\left(
        \frac{
            \alpha_{q}t^{q+1}
        }{
            r^{q}MK
        }
    \right)
    =
    O\!\left(
        \frac{
            \alpha_{q}t^{q+1}
        }{
            r^{q+1}K
            \nu_q|\mathfrak{X}|
        }
    \right).
    \label{eq:delta_round_trotter_scaling}
\end{equation}
Alternatively, if $r$ is chosen so that
$\epsilon_{\rm PF}(t,r)\le\epsilon/2$, one may take
$\varepsilon_{\rm res}\le\epsilon/2$, giving
\begin{equation}
    \delta_{g,j}
    =
    O\!\left(
        \frac{\epsilon}{MK}
    \right)
    =
    O\!\left(
        \frac{\epsilon}{
            rK \nu_q|\mathfrak{X}|
        }
    \right).
    \label{eq:delta_round_target_xrror}
\end{equation}

Finally, we translate each $\delta_{g,j}$into an explicit finite-Bell-pair
preparation cost. For round $j$ of attempt $g$, let
$
    \alpha_{g,j}=2^j\theta_g
$ and $
    p_{g,j}=\sin^2\alpha_{g,j}.
$
If \(n\) copies of the resource with angle \(\alpha_{g,j}\) are prepared as a
block, Proposition~\ref{thm:finite_dilution_truncation} shows that it is
sufficient to choose \(L_{g,j}\) such that
\begin{equation}
    \Pr\!\left[
        {\rm Bin}(n,p_{g,j})>L_{g,j}
    \right]
    \le
    \delta_{g,j}^2 .
    \label{eq:round_tail_condition}
\end{equation}
The corresponding integer number of Bell pairs is
\begin{equation}
    m_{g,j}
    =
    \left\lceil
        \log_2
        \sum_{w=0}^{L_{g,j}}
        \binom{n}{w}
    \right\rceil .
    \label{eq:round_integer_cost}
\end{equation}
Using the Bernstein-Chernoff bound, a sufficient explicit choice is
\begin{equation}
    L_{g,j}
    =
    \left\lceil
        np_{g,j}
        +
        \sqrt{
            2np_{g,j}
            \log\frac{1}{\delta_{g,j}^2}
        }
        +
        \frac{2}{3}
        \log\frac{1}{\delta_{g,j}^2}
    \right\rceil .
    \label{eq:Lgj_choice}
\end{equation}
Together with
$
    \sum_{w=0}^{L}\binom{n}{w}
    \le
    \left(\frac{en}{L}\right)^L
 $ and $
    1\le L\le n,$
this gives the finite-copy scaling
\begin{equation}
    m_{g,j}
    =
    O\!\left[
        L_{g,j}
        \log_2\!\left(
            \frac{en}{L_{g,j}}
        \right)
    \right].
    \label{eq:m_scaling_round}
\end{equation}
Thus finite resource synthesis only needs to meet the per-round accuracy specified in \eqref{eq:delta_round_target_xrror}. Exact preparation is unnecessary: the
finite-preparation error is absorbed into the total simulation error budget.
The possibility of repeated unfavorable RUS outcomes is already accounted for
by the package condition
$\sum_{j=0}^{K-1}\delta_{g,j}\le\varepsilon_{\rm res}/M$.

\end{widetext}

\end{document}